\documentclass[11pt,onecolumn,oneside,aps,pra,tightenlines,notitlepage,nofootinbib,superscriptaddress,floatfix]{revtex4-2}

\usepackage{amssymb}
\usepackage{amsmath}
\usepackage{amsbsy}
\usepackage{amsfonts}
\usepackage{mathrsfs}
\usepackage{amsthm}
\usepackage{mathtools}

\usepackage{graphicx}

\usepackage{tikz}
\usepackage{tikz-cd}
\usetikzlibrary{quantikz2}

\usepackage{xcolor}
\definecolor{darkred}  {rgb}{0.5,0,0}
\definecolor{darkblue} {rgb}{0,0,0.5}
\definecolor{darkgreen}{rgb}{0,0.5,0}

\usepackage{hyperref}
\hypersetup{
  breaklinks = true,
  colorlinks = true,
  urlcolor = blue,
  linkcolor = darkblue,
  citecolor = darkgreen,
  filecolor = darkred
}

\theoremstyle{plain}
\newtheorem{theorem}{Theorem}

\newtheorem{lemma}{Lemma}
\newtheorem{corollary}{Corollary}

\theoremstyle{definition}

\newtheorem{remark}{Remark}

\renewcommand{\leq}{\leqslant}

\newcommand{\Nat}{\mathbb{N}}
\newcommand{\Int}{\mathbb{Z}}
\newcommand{\Real}{\mathbb{R}}
\newcommand{\Comp}{\mathbb{C}}

\DeclareMathOperator{\Tr}{Tr}

\DeclarePairedDelimiter{\abs}{\lvert}{\rvert}

\DeclarePairedDelimiterX{\innerp}[2]{\langle}{\rangle}{#1, #2}
\DeclarePairedDelimiter{\bra}{\langle}{\rvert}
\DeclarePairedDelimiter{\ket}{\lvert}{\rangle}
\DeclarePairedDelimiterX{\braket}[2]{\langle}{\rangle}{#1 \delimsize\vert #2}
\DeclarePairedDelimiterX{\ketbra}[2]{\lvert}{\rvert}{#1 \delimsize\rangle\!\delimsize\langle #2}
\DeclarePairedDelimiterX{\proj}[1]{\lvert}{\rvert}{#1 \delimsize\rangle\!\delimsize\langle #1}

\newcommand{\NOT}{\mathrm{NOT}}
\newcommand{\CNOT}{C\NOT}

\newcommand{\Hil}{\mathcal{H}}
\newcommand{\Bounded}{\mathfrak{B}}
\newcommand{\TrClass}{\mathfrak{T}}
\newcommand{\Unitary}{\mathfrak{U}}

\newcommand{\Id}{\mathrm{Id}}
\newcommand{\Deph}{\mathcal{D}}

\newcommand{\PauliGroup}{\mathcal{P}}
\newcommand{\CliffordGroup}{\mathcal{C}\!\ell}
\newcommand{\PStab}{\mathrm{PStab}}
\newcommand{\Stab}{\mathrm{Stab}}
\newcommand{\MixStab}{\mathrm{MixStab}}

\begin{document}

\title{Characterization of non-adaptive Clifford channels}

\author{Vsevolod I. Yashin}
\email{yashin.vi@mi-ras.ru}
\affiliation{Steklov Mathematical Institute of Russian Academy of Sciences, Moscow 119991, Russia}
\affiliation{Russian Quantum Center, Skolkovo, Moscow 143025, Russia}

\author{Maria A. Elovenkova}
\affiliation{Steklov Mathematical Institute of Russian Academy of Sciences, Moscow 119991, Russia}
\affiliation{Moscow Institute of Physics and Technology, Dolgoprudny 141700, Russia}

\date{\today}

\begin{abstract}
  Stabilizer circuits arise in almost every area of quantum computation and communication, so there is interest in studying them from an information-theoretic perspective, i.e.\ as quantum channels. We consider several natural approaches to what can be called a Clifford channel: the channel that can be realised by a stabilizer circuit without classical control, the channel that sends pure stabilizer states to mixed stabilizer states, the channel with stabilizer Choi state, the channel whose Stinespring dilation can have a Clifford unitary. We show the equivalence of these definitions. Up to unitary encoding and decoding maps any Clifford channel is a product of stabilizer state preparations, qubit discardings, identity channels and full dephasing channels. This simple structure allows to compute information capacities of such channels.
\end{abstract}

\maketitle

\section{Introduction} \label{sec:introduction}

Various protocols of quantum information processing are described by the use of quantum circuits \cite{Holevo_2019, Wilde_2017}. Most of such protocols extensively use stabilizer operations: introduction of qubits and preparations of the initial state $\ket{0}$, Clifford group gates $\{H,S,\CNOT\}$, discarding of qubits, single qubit measurements and classical data processing. Quantum circuits constructed only of Clifford operations are called stabilizer circuits \cite{Gottesman_1997}. This class of circuits is widely used in quantum computation theory, especially in error correction. It is therefore valuable to understand the properties of information transmission through stabilizer circuits.

A lot of stress in quantum Shannon theory is devoted to the study of Gaussian channels \cite{Weedbrook_2012, Eisert_2007, Giovannetti_2014}. There is an analogy between stabilizer circuits and linear optical circuits: stabilizer states can be understood as discrete analogues of Gaussian states, Clifford gates as discrete analogues of linear optical gates, qubit measurements as homodyne measurements of bosonic modes. That analogy becomes even more natural for odd-dimensional qudit stabilizer circuits, where there exists a developed theory of discrete Wigner functions \cite{Gross_2006, Gross_2006_2, Bu_2023}. Therefore studying stabilizer channels can provide insight into the structure of Gaussian channels.

A central result in stabilizer circuits theory is the Gottesman-Knill theorem \cite{Gottesman_1997, Gottesman_1998, Aaronson_2004}, which establishes that such circuits are efficiently classically simulated, thereby precluding any computational advantage over classical circuits. The strength of the simulation depends on the availability of classical control in the circuit \cite{Aaronson_2004, Van_den_Nest_2010}. Stabilizer circuits with allowed classical control (adaptive circuits) can be simulated weakly, i.e.\ one can efficiently sample measurement outcomes. At the same time, for circuits without adaptivity it is also possible to perform strong simulation, i.e.\ to explicitly find probabilities of chosen outcomes \cite{Jozsa_2013, Koh_2017}.

Recently, a number of works \cite{Bennink_2017, Seddon_2019, Seddon_2021, Howard_2017, Heimendahl_2022} have considered quantum channels representable by adaptive stabilizer circuits, which we call adaptive Clifford channels. As mentioned, such channels can be efficiently weakly simulated. Also, one can decompose any quantum channel as a sum of adaptive Clifford channels, allowing for weak simulation of arbitrary quantum circuits. The difficulty of the simulation depends on the amount of magic (also called non-stabilizerness) in the circuit \cite{Bravyi_2005, Bravyi_2016, Bravyi_2016_2, Bravyi_2019, Seddon_2021}. However, the careful analysis of stabilizer circuits without classical control was lacking. And, generally, there seems to be some confusion in various definitions and namings for classes of channels realised by stabilizer circuits.

In this work, we endeavour to narrow this knowledge gap and to put things in order. In essence, we show that channels described by stabilizer circuits without classical control (or with restricted class of, see Remark~\ref{rmk:allowed_control}), which we call Clifford channels, have especially simple structure and information-theoretic properties. Most importantly, we carefully prove that several natural definitions equivalently define Clifford channels. It can be summarised in the following Theorem:
\begin{theorem}[Main result] \label{thm:main_result}
  Let $\Phi : \TrClass(\Hil_A) \to \TrClass(\Hil_B)$ be a multiqubit quantum channel. The following statements are equivalent:
  \begin{enumerate}
    \item $\Phi$ can be realized by a stabilizer circuit without classical control.
    \item $\Phi$ sends pure stabilizer states to mixed stabilizer states.
    \item $\Id_E\otimes\Phi$ for any qubit environment $E$ sends mixed stabilizer states to mixed stabilizer states.
    \item Choi state $\sigma_\Phi$ of the channel $\Phi$ is mixed stabilizer.
    \item $\Phi$ admits a Stinespring dilation with a Clifford unitary.
    \item Dual channel $\Phi^*$ maps Pauli observables to Pauli obvervables or zero.
  \end{enumerate}
  If any of these properties is satisfied, we call $\Phi$ a Clifford channel.
\end{theorem}

\noindent The proof of this result is distributed to a number of Theorems and Lemmata throughout the paper. Knowing this characterization, one finds that up to Clifford unitary equivalence any Clifford channel is constructed of: identity channels, dephasing channels, state $\ket{0}$ and chaotic state $\chi$ preparations, and qubit discardings \cite{Looi_2011}. The results also allow to better understand the structure of classical-to-quantum stabilizer state preparation channels, Pauli measurement channels and affine classical channels.

The paper is structured as follows. In Section~\ref{sec:preliminaries} we establish the notation and give necessary preliminaries about stabilizer circuits. In Section~\ref{sec:characterization} we prove some characterization theorems for classes of unitary, isometric and arbitrary non-adaptive Clifford channels. In Section~\ref{sec:normal_form} we discuss how any Clifford channel can be reduced to a certain normal form, which makes it possible to compute its information capacities. In Section~\ref{sec:stabilizer_preserving} we further investigate stabilizer preserving properties of Clifford channels and complete the proof of the main result. Finally, in Section~\ref{sec:conlusion} we give a brief conclusion and elaborate on possible future research.

\section{Preliminaries and notation} \label{sec:preliminaries}

In this Section, we fix the notation and briefly recall some results of quantum channels \cite{Holevo_2019, Wilde_2017} and stabilizer circuits theory \cite{Gottesman_1997, Gottesman_1998, Aaronson_2004, Garcia_2014}.

\subsection{Quantum channels}

We will denote $\Hil$ or $\Hil_A$ Hilbert spaces over a quantum system $A$, $\TrClass(\Hil)$ the space of trace-class operators on $\Hil$ and $\Bounded(\Hil)$ the space of bounded operators on $\Hil$. We will denote the group of unitaries over $\Hil$ as $\Unitary(\Hil)$. By a \emph{quantum channel} $\Phi : \TrClass(\Hil_A) \to \TrClass(\Hil_B)$ we mean a completely positive trace-preserving linear map between systems $A$ and $B$, we can also denote it $\Phi_{A\to B}$ for clarity. For each channel $\Phi$ one defines a \emph{dual channel} $\Phi^* : \Bounded(\Hil_B) \to \Bounded(\Hil_A)$ by the relation
\begin{equation}
  \Tr(\Phi[\rho] B) = \Tr(\rho \, \Phi^*[B])
\end{equation}
for all states $\rho\in\TrClass(\Hil_A)$ and observables $B\in\Bounded(\Hil_B)$. Dual channel $\Phi^*$ is a unit preserving completely positive map.

We denote the identity channel as $\Id : \TrClass(\Hil) \to \TrClass(\Hil)$. \emph{Unitary channels} are the channels of the form $\Phi[\rho] = U \rho U^\dag$, where $U\in\Unitary(\Hil)$ is a unitary. We call \emph{isometry channels} the maps of the form $\Phi[\rho] = V \rho V^\dag$, where $V : \Hil_A \to \Hil_B$ is some isometry $V^\dag V = I_A$ which embeds $\Hil_A$ into a larger Hilbert space $\Hil_B$. When $\Hil_B = \Hil_A\otimes\Hil_E$, each isometry $V : \Hil_A \to \Hil_B$ can be expressed as
\begin{equation}
  V \ket{\psi}_A = U\, \ket{\psi}_A\ket{0}_E
\end{equation}
for some unitary $U\in\Unitary(\Hil_B)$. By Stinespring's dilation theorem, each quantum channel $\Phi : \TrClass(\Hil_A) \to \TrClass(\Hil_B)$ can be purified to the form $\Phi[\rho] = \Tr_E(V\rho V^\dag)$ for some isometry $V : \Hil_A \to \Hil_{BE}$.

Given a superoperator $\Phi : \TrClass(\Hil_A) \to \TrClass(\Hil_B)$ and the maximally entangled pure state $\Omega_{A'A} \in \TrClass(\Hil_{A'A})$ on a bipartite system $A'|A$, where $A'$ is a copy of $A$, the \emph{Choi state} $\sigma$ is defined as
\begin{equation}
  \sigma_{A'B} = \Id_{A'}\otimes\Phi_{A\to B}[\Omega_{A'A}].
\end{equation}
It is an operator on the bipartite system $A'|B$. \emph{Choi-Jamio{\l}kowsky duality} states that the superoperator $\Phi$ is completely positive if and only if $\sigma$ is positive, and $\Phi$ is trace-preserving if and only if $\Tr_B\sigma$ is the maximally mixed state on $A'$.

\subsection{Quantum circuits}

Hilbert spaces of $n$-qubit systems will be denoted $\Hil^n = (\Comp^2)^{\otimes n}$. More generally, the upper index will denote $n$-th tensor power and $\abs{A} = n$ will mean that $A$ is an $n$-qubit system.\footnote{Even though the results hold for qudit systems of any prime local dimensions, we restrict to qubit case, which is the most important.} We will denote the computational basis as $\{\ket{x}\}_x$ where $x\in\Int_2^n$ are bit strings. States on a system of $n$ classical bits are represented by probability distributions from $\ell_1(\Int_2^n)$, and can be understood as diagonal operators in $\TrClass(\Hil^n)$.

By a \emph{quantum circuit} we mean some diagram assembled of qubit wires, clasical bit wires, and a sequence of elementary operations on them. Let us list these operations. First of all, one can introduce additional qubits and prepare them to initial state $\ket{0}$. Also, we allow preparing \emph{maximally mixed (chaotic) states} $\chi = \frac{1}{2} I$. One can act on multiple wires by multiqubit gates $U$ from some dictionary. We will consider the operation of discarding qubits $\rho \mapsto \Tr(\rho)$, which we also call \emph{erasure channel} $\Tr$. It can be treated as the unique channel into $0$-qubit system $\Hil^0 = \Comp$. We suppose a qubit can be (destructively) measured in $Z$-basis to obtain one classical bit. We will call a circuit \emph{unitary} if it consists only of unitary gates, and \emph{isomety circuit} if it consists of initial state $\ket{0}$ preparations and unitary gates.

The classical data can be processed via classical operations: preparations of classical bits in states $0$ and $1$ or uniformly random bits, actions of Boolean gates, bit discardings and classical readouts. A classical circuit is called \emph{affine} \cite{Patel_2003} if it consists of only Boolean gates describing addition, such as $\NOT$ and $\CNOT$. It is possible to change the behaviour of qubit operations via classical control depending on bit state. Circuits with classical control are called \emph{adaptive} \cite{Terhal_2004}. We will mainly be interested in \emph{non-adaptive} circuits without any classical processing. At the same time, one can get rid of adaptivity using \emph{deferred measurement principle} \cite{Nielsen_2010} (also see Remark~\ref{rmk:allowed_control}):
\begin{equation}
  \begin{quantikz}[wire types={q,q}, row sep={1.0cm,between origins}, column sep={0.8cm,between origins}, align equals at=1.5]
    &\meter{} &[0.2cm]\ctrl{0}\wire[d][1]{c}\setwiretype{c} & \\
    &\meter{} &\targ{}\setwiretype{c}                       &
  \end{quantikz}
  \quad = \quad
  \begin{quantikz}[wire types={q,q}, row sep={1.0cm,between origins}, column sep={0.8cm,between origins}, align equals at=1.5]
    &\meter{} &[0.2cm]\ctrl{0}\wire[d][1]{c}\setwiretype{c} &[0.2cm]  &                \\
    &           &\gate{X}                                     &\meter{} &\setwiretype{c}
  \end{quantikz}
  \quad = \quad
  \begin{quantikz}[wire types={q,q}, row sep={1.0cm,between origins}, column sep={0.8cm,between origins}, align equals at=1.5]
    &\ctrl{1} &[0.2cm]\meter{} &\setwiretype{c} \\
    &\targ{}  &\meter{}        &\setwiretype{c}
  \end{quantikz}
\end{equation}

Other possible circuit elements can be composed of the listed operations. Action of a circuit on inputs is described by some quantum channel $\Phi$. In that case a quantum channel $\Phi$ is \emph{realised} (or \emph{represented}) by this circuit. Circuits representing the same channel are called \emph{equivalent}. By Stinespring's dilation theorem, any channel $\Phi : \TrClass(\Hil_A) \to \TrClass(\Hil_B)$ can be represented by some circuit of form
\begin{equation}
  \begin{quantikz}[wire types={q}, column sep={0.6cm,between origins}, align equals at=1.0]
    \lstick{$A$} &\gate{\Phi} &\rstick{$B$}
  \end{quantikz}
  \; = \;
  \begin{quantikz}[wire types={n,q}, row sep={0.8cm,between origins}, column sep={0.65cm,between origins}, align equals at=1.5]
    \lstick{$A'$} &\lstick{$\ket{0}$} &[0.4cm]\gate[2][1.5cm]{U}\setwiretype{q} &[0.5cm]\ground{} &\rstick{$B'$}\setwiretype{n} \\
    \lstick{$A$}  &                   &                                         &                 &\rstick{$B$}
  \end{quantikz}
  ,
\end{equation}
where $A'$, $B'$ are some environments and $U : \Hil_{AA'}\to\Hil_{BB'}$ is a unitary operation.

\subsection{Stabilizer formalism}

An $n$-qubit \emph{Pauli group} $\PauliGroup^n$ is a group of unitary operators which are tensor products of one-qubit Pauli matrices $I,X,Y,Z$ with possible phase $i$:
\begin{equation}
  \PauliGroup^n = \langle iI, Z_1, X_1, \dots, Z_n, X_n \rangle.
\end{equation}
Hermitian elements of Pauli group are called \emph{Pauli observables}. We will call a Pauli observable $P\in\PauliGroup^n$ \emph{non-trivial} if $P\neq \pm I$. The set of non-signed Pauli matrices constitutes an orthonormal basis in the space of matrices with Hilbert-Schmidt inner product $(X,Y)\mapsto \frac{1}{2^n}\Tr(X^\dag Y)$.

Unitaries that preserve Pauli group are called \emph{Clifford}, and the set of such unitaries is called \emph{Clifford group} $\CliffordGroup^n$:
\begin{equation}
  \CliffordGroup^n = \{ U\in\Unitary(\Hil^n) : \text{if }P\in\PauliGroup^n\text{ then }U P U^\dag\in\PauliGroup^n \}.
\end{equation}
Clifford group is generated by one-qubit gates $\{S, H\}$ and two-qubit gates $\CNOT$, while \emph{diagonal} Clifford unitaries are generated by gates $S$ and $CZ$. The transposition $U^T$ of a Clifford unitary $U\in\CliffordGroup^n$ is again Clifford. The structure of Pauli and Clifford groups is closely related to symplectic geometry of $\Int_2^{2n}$ \cite{Calderbank_1997, Calderbank_1998, Gross_2006}, but we will not employ such results in this work.

A \emph{stabilizer group} $\mathcal{S} \subset \PauliGroup^n$ is a commutative group of Pauli observables such that $-I \notin \mathcal{S}$. Any stabilizer group is isomorphic to $\Int_2^k$, where $k$ is a \emph{rank} of $\mathcal{S}$ and $0\leq k\leq n$. Therefore any stabilizer group $\mathcal{S}$ has a set of \emph{generators} $\langle P_1,\dots,P_k\rangle$. The set of generators (i.e.\ signed Pauli strings of length $n$) written into a column is called a \emph{stabilizer tableau}. Stabilizer tableau is defined up to row operations: permutation of generators and multiplication of one generator by another.

If a stabilizer group $\mathcal{S}$ is maximal (i.e.\ $k=n$), then it defines a unique \emph{pure stabilizer state} $\ket{\psi}\in\Hil^n$ such that $P\ket{\psi}=\ket{\psi}$ for all $P\in\mathcal{S}$.\footnote{Note that we define pure stabilizer states and Clifford unitaries up to any global phase, while for some applications it is crucial to consider $\Int_8$ phases \cite{Bravyi_2016_2, Bravyi_2016}.} We will denote the set of $n$-qubit pure stabilizer states as $\PStab^n$. Important examples of pure stabilizer states include $Z$-basis states $\{\ket{0},\ket{1}\}$, $X$-basis states $\{\ket{+},\ket{-}\}$, maximally entangled Bell state $\ket{\mathrm{Bell}}$ and Greenberger–Horne–Zeilinger (GHZ) states $\ket{\mathrm{GHZ}}$:
\begin{equation}
  \ket{\mathrm{Bell}} = \frac{\ket{00}+\ket{11}}{\sqrt{2}}, \qquad
  \ket{\mathrm{GHZ}} = \frac{\ket{000}+\ket{111}}{\sqrt{2}}.
\end{equation}
We will also denote the maximally entangled state as $\Omega = \proj{\mathrm{Bell}}$. Any pure $n$-qubit stabilizer states has form $\ket{\psi} = U\ket{0^n}$ for some $U\in\CliffordGroup^n$. To any stabilizer state $\ket{\psi}$ corresponds some stabilizer tableau. For example, to the mentioned stabilizer states correspond the tableaux:
\begin{equation}
  \ket{0^n} \leadsto
  \begin{bmatrix}
    Z & I & \cdots & I  \\
	I & Z & \cdots & I   \\
	\vdots & \vdots & \ddots & \vdots \\
	I & I & \cdots & Z
  \end{bmatrix}
  ,
  \qquad
  \ket{\mathrm{Bell}} \leadsto
  \begin{bmatrix}
    Z & Z \\
    X & X
  \end{bmatrix}
  ,
  \qquad
  \ket{\mathrm{GHZ}} \leadsto
  \begin{bmatrix}
    Z & Z & I \\
    I & Z & Z \\
    X & X & X
  \end{bmatrix}
  .
\end{equation}


More generally, arbitrary stabilizer group $\mathcal{S}$ with generators $\{P_1,\dots,P_k\}$ defines a \emph{mixed stabilizer state} $\rho \in \TrClass(\Hil^n)$ \cite{Aaronson_2004, Fattal_2004, Audenaert_2005} by the formula
\begin{equation} \label{eq:mixed_stabilizer}
  \rho = \frac{1}{2^n} \sum_{P\in\mathcal{S}} P = \frac{1}{2^n} (I+P_1)\cdots(I+P_k).
\end{equation}
An important example is the mixed version of the Bell state, which is the partial trace of $\ket{\mathrm{GHZ}}$ by third qubit:
\begin{equation}
  \Sigma
  = \Tr_3\proj{\mathrm{GHZ}}
  = \frac{\proj{00} + \proj{11}}{2}
  \quad \leadsto \quad
  \begin{bmatrix}
    Z & Z
  \end{bmatrix}
  .
\end{equation}
We will call this state \emph{maximally correlated state} for it exhibits classical correlations but is not entangled. Any \emph{classical (diagonal)} mixed stabilizer state corresponds to a uniform distribution over some $k$-dimensional affine subspace $L$ of $\Int_2^n$:
\begin{equation}
  \rho = \frac{1}{2^k} \sum_{x\in L} \proj{x}.
\end{equation}
We will denote the set of $n$-qubit mixed stabilizer states as $\Stab^n$, and if we say a state $\rho$ is stabilizer, we mean it can be mixed. Mixed stabilizer states can be written as $\rho = U (\proj{0}^k\otimes\chi^{n-k}) U^\dag$ for some Clifford unitary $U\in\CliffordGroup^n$. It follows that mixed stabilizer states can be purified $\rho = \Tr_E\proj{\psi}$ to a pure stabilizer state $\ket{\psi}\in\PStab^{2n-k}$, and such minimal purifications are Clifford unitary equivalent.
It is important not to confuse mixed stabilizer states with \emph{mixtures} of stabilizer states, which are arbitrary convex combinations of pure stabilizer states.

The class of mixed stabilizer states is closed under tensor products and partial traces. Let us consider a bipartite qubit system $A|B$. Suppose $\rho\in\Stab_{AB}$ has stabilizer group $\mathcal{S}_{AB}$. Then the partial trace over $B$ gives
\begin{equation}
  \rho_A
  = \Tr_B\rho_{AB} \\
  = \frac{1}{2^n} \sum_{P_{AB}\in\mathcal{S}_{AB}} \Tr_B P_{AB} \\
  = \frac{1}{2^n} \sum_{P_A\in\mathcal{S}_A} P_A,
\end{equation}
where $\mathcal{S}_A\subset \PauliGroup_A$ is a stabilizer group corresponding to Pauli observables from $\mathcal{S}_{AB}$ that are trivial on $B$. It was shown in \cite{Fattal_2004} that any pure stabilizer state on a bipartite system can be brought to a product of local pure stabilizer states $\ket{0}$ and some number of Bell states $\ket{\mathrm{Bell}}$ (EPR pairs) using local Clifford unitaries. Also, in \cite{Bravyi_2006} it was shown that three-party pure stabilizer states are locally Clifford equivalent to a product of local states, Bell states between pairs of parties, and some number of three-qubit GHZ states $\ket{\mathrm{GHZ}}$. These results are used to characterise bipartite mixed stabilizer states, see Theorem~\ref{thm:bipartite_mixstab_states}.

\subsection{Stabilizer circuits}

By a \emph{non-adaptive stabilizer circuit} we understand a quantum circuit which is constructed of classical and Clifford elements: initial state $\ket{0}$ and chaotic state $\chi$ preparations, local Clifford gates, qubit discardings, local Pauli measurements, arbitrary classical Boolean operations. If we also allow classical control over Clifford elements, then we call a circuit \emph{adaptive}. Additionally, we will consider a one-qubit \emph{dephasing gate} $\Deph[\rho] = \frac{1}{2}(\rho + Z\rho Z)$, which can be expressed via stabilizer circuit as
\begin{equation}
  \begin{quantikz}[wire types={q}, column sep={0.8cm,between origins}, align equals at=1.0]
    &\gate{\Deph} &
  \end{quantikz}
  \quad = \quad
  \begin{quantikz}[wire types={n,q}, row sep={0.8cm,between origins}, column sep={0.8cm,between origins}, align equals at=1.5]
    &\lstick{$\ket{0}$} &\targ{}\setwiretype{q} &\ground{} &\setwiretype{n} \\
    &                   &\ctrl{-1}              &          &
  \end{quantikz}
  .
\end{equation}

Stabilizer circuits are widely used in error correction and are known to be classically simulable. More exactly, for adaptive stabilizer circuits the problem of \emph{weak simulation}, i.e.\ the problem of sampling from the outcome distribution, is efficiently solvable on classical computers using various approaches: stabilizer tableau methods \cite{Gottesman_1997, Gottesman_1998, Aaronson_2004, Gidney_2021}, graph state representations \cite{Anders_2006, Rijlaarsdam_2020}, quadratic form expansions \cite{Dehaene_2003, Van_den_Nest_2010, Bravyi_2016, Beaudrap_2022}, Wigner functions and other quasiprobability representations \cite{Gross_2006, Veitch_2012, Raussendorf_2020, Park_2023, Pashayan_2015, Delfosse_2015, Kulikov_2024}, hidden variable models \cite{Raussendorf_2017, Zurel_2020}. \emph{Strong simulation}, i.e.\ the problem of computing outcome probabilities, is computationally hard for general classical circuits and therefore for adaptive stabilizer circuits \cite{Van_den_Nest_2010, Jozsa_2013, Koh_2017}. However, for \emph{non-adaptive} Clifford circuits the strong simulation is efficient (also see Remark~\ref{rmk:allowed_control}), and this problem is reducible to computations on \emph{affine} Boolean circuits \cite{Aaronson_2004, Buhrman_2006}.

\section{Characterization of Clifford channels} \label{sec:characterization}

In this Section we discuss various definitions of Clifford channels: in terms of stabilizer preserving properties, Choi states and stabilizer circuit representability, etc. We then prove that part of these definitions are equivalent. Completing the characterization is left for Section~\ref{sec:stabilizer_preserving}.

\subsection{Definitions of Clifford channels}

As a starting point, let us agree to call \emph{Clifford channels} the channels representable by non-adaptive stabilizer circuits. Channels representable by adaptive stabilizer circuits we will call \emph{adaptive Clifford channels}. A simple example of a channel which is adaptive Clifford but not Clifford is given by stabilizer circuit
\begin{equation} \label{eq:example_of_adaptive}
  \begin{quantikz}[wire types={n,q}, row sep={1.0cm,between origins}, column sep={0.6cm,between origins}, align equals at=1.5]
    &\lstick{$\ket{0}$} &\targ{}\setwiretype{q} &[0.2cm]\meter{} &[0.2cm]\ctrl{0}\wire[d][1]{c}\setwiretype{c} &[0.4cm]\setwiretype{n} \\
    &                   &\ctrl{-1}              &                &\gate{H}                                     &
  \end{quantikz}
  .
\end{equation}
Another example is the non-destructive measurement of two-qubit observable $Z_1 + Z_2$, i.e.\ total angular momentum of two $1/2$-spins. It cannot be accomplished by measuring independent commuting Pauli observables, but it can be realised by first measuring parity observable $Z_1Z_2$, then measuring one of qubits if the outcome is $0$, or do nothing for outcome $1$.

In this work we will consider properties of non-adaptive Clifford channels. The structure of adaptive channels is much more complicated, for deeper discussion consult \cite{Bennink_2017, Seddon_2019, Heimendahl_2022}.
\begin{remark} \label{rmk:allowed_control}
  A restricted amount of adaptivity for stabilizer circuits can be allowed to still result in Clifford channels. More concretely, we can allow affine Boolean operations over classical wires with control over Pauli gates. Any stabilizer circuit of this class can be rewritten to non-adaptive stabilizer circuit: classical wires can be represented as \emph{dephased} quantum wires (e.g.\ insert dephasing gate $\Deph$ at each step); affine Boolean operations can be expressed by $\NOT$ and $\CNOT$ quantum gates; applying deferred measurement principle to classical control over Pauli gates results in quantum controlled-Pauli gates, which are Clifford. The procedure is very similar to the one used to construct coherent communication protocols \cite[Chapter 7]{Wilde_2017}. The importance of this class of circuits is also discussed in \cite{Kliuchnikov_2023}.
\end{remark}

We will say that a channel $\Phi$ has a \emph{Clifford dilation} \cite{Heimendahl_2022} if it has a Stinespring dilation with Clifford unitary part. Any Clifford dilation can be realised by a stabilizer circuit, so such channels are Clifford. The converse will be proven in this Section. Note that the channel in Equation~(\ref{eq:example_of_adaptive}) can as well be represented as
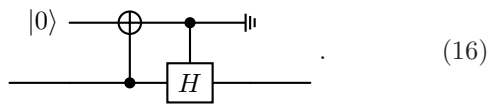
\begin{equation}
  \begin{quantikz}[wire types={n,q}, row sep={0.8cm,between origins}, column sep={0.8cm,between origins}, align equals at=1.5]
    &\lstick{$\ket{0}$} &\targ{}\setwiretype{q} &\ctrl{1} &\ground{} &\setwiretype{n} \\
    &                   &\ctrl{-1}              &\gate{H} &          &
  \end{quantikz}
  .
\end{equation}
The unitary part of this circuit lies in the $3$-rd level of Clifford hierarchy \cite{Gottesman_1999}. This suggests a definition to \emph{channels of $l$-th level in Clifford hierarchy} as channels having a unitary of $l$-th level in Stinespring dilation. It could be awarding to inverstigate the connection between adaptivity and Clifford hierarchy.

Another natural possibility to define Clifford channels is to rely on their stabilizer preserving properties.\footnote{It is quite often that stabilizer operations are defined as channels that map pure stabilizer states to pure stabilizer states (e.g.\ \cite{Gu_2023}), but we have not come across a careful justification for such definition.} We will call a channel $\Phi : \TrClass(\Hil_A) \to \TrClass(\Hil_B)$ \emph{stabilizer preserving} if it maps all pure stabilizer states $\ket{\psi}\in\PStab_A$ to mixed stabilizer states $\Phi[\proj{\psi}]\in\Stab_B$. Surely, channels representable by non-adaptive stabilizer circuits have this property. The channel $\Phi$ is \emph{completely stabilizer preserving} if for any qubit environment $E$ the amplification $\Id_E\otimes\Phi$ is stabilizer preserving. We will call a channel $\Phi$ \emph{pure stabilizer preserving} if it maps pure stabilizer states to pure stabilizer states, and \emph{completely pure stabilizer preserving} if this condition holds for all $\Id_E\otimes\Phi$. Also, we will call $\Phi$ \emph{adaptive stabilizer preserving} if it maps pure stabilizer states to mixtures of stabilizer states, and \emph{completely adaptive stabilizer preserving} likewise.

Let us compare these definitions. It is straightforward that pure stabilizer preserving channels are stabilizer preserving, which themselves are adaptive stabilizer preserving. If a channel is completely stabilizer preserving, then it also sends all mixed stabilizer states to mixed stabilizer states (consider the action on purifications). It is not obvious that usual stabilizer preserving channels map mixed stabilizer states to mixed stabilizer states. That does not hold for classical channels: the non-linear map $(x,y)\mapsto x\cdot y$ is deterministic, but maps the uniform distribution over $\Int_2^2$ to non-uniform distribution over $\Int_2$. Also, it was shown in \cite[Section 4]{Seddon_2019} that adaptive stabilizer preserving channels are not necessarily completely adaptive stabilizer preserving. Nevertheless, in Theorem~\ref{thm:stabilizer_preserving_are_clifford} we prove that stabilizer preserving channels are always completely stabilizer preserving.

Finally, since the properties of a channel $\Phi$ are tied to properties of it's Choi state $\sigma$, it would be natural to expect that Choi states of Clifford channels are stabilizer. The list of elementary operations and their respective Choi states is given in Table~\ref{tab:Choi_states}, for these elements the Choi states are indeed stabilizer.

\begin{table}[h]
  \caption{The table of elementary stabilizer operations from $A$ to $B$ and their Choi states on system $A'|B$.}
  \label{tab:Choi_states}
  \begin{ruledtabular}
  \begin{tabular}{ccc}
  Channel $\Phi_{A\to B}$& Circuit depiction & Choi state $\sigma_{A'B}$ \\
  \hline
  initial state $\ket{0}$ preparation & \begin{quantikz} \lstick{$\ket{0}$} & \end{quantikz}               & $\proj{0}_B$ \\
  chaotic state $\chi$ preparation    & \begin{quantikz} \lstick{$\chi$} & \end{quantikz}                  & $\chi_B$ \\
  erasure channel $\Tr$               & \begin{quantikz} &\ground{} \end{quantikz}                         & $\chi_{A'}$ \\
  identity channel $\Id$              & \begin{quantikz} & & \end{quantikz}                                & $\Omega_{A'B}$ \\
  dephasing channel $\Deph$           & \begin{quantikz}[column sep={0.3cm}] &\gate{\Deph} & \end{quantikz} & $\Sigma_{A'B}$ \\
  \end{tabular}
  \end{ruledtabular}
\end{table}

\subsection{Characterization theorems}

Let us prove that some of the introduced notions are equivalent. Similar characterization is well known for bosonic Gaussian channels \cite{Giedke_2002, Caruso_2008, Caruso_2011, Weedbrook_2012}. Here we use the notion of completely stabilizer preserving channels, but in Section~\ref{sec:stabilizer_preserving} we prove that the usual stabilizer preserving property suffices.

Let us start by characterising unitary channels.

\begin{theorem}[Unitary Clifford channels] \label{thm:unitary_characterization}
  Let $\Phi : \TrClass(\Hil_A) \to \TrClass(\Hil_A)$ be a channel. The following statements are equivalent:
  \begin{enumerate}
    \item The channel $\Phi$ is completely pure stabilizer preserving.
    \item The Choi state $\sigma$ is maximally entangled and pure stabilizer.
    \item The channel $\Phi$ may be realised by a unitary stabilizer circuit.
  \end{enumerate}
\end{theorem}

\noindent The theorem for unitary channels is a special case of isometry channels characterization:

\begin{theorem}[Isometry Clifford channels] \label{thm:isometry_characterization}
  Let $\Phi : \TrClass(\Hil_A) \to \TrClass(\Hil_B), \abs{A}\leq\abs{B}$ be a channel. The following statements are equivalent:
  \begin{enumerate}
    \item $\Phi$ is completely pure stabilizer preserving.
    \item Choi state $\sigma$ is pure stabilizer.
    \item $\Phi$ may be realised by an isometry stabilizer circuit.
  \end{enumerate}
\end{theorem}

\begin{proof}
  The directions $\emph{1}\Rightarrow\!\emph{2}$ and $\emph{3}\Rightarrow\!\emph{1}$ are straightforward. Let us prove $\emph{2}\Rightarrow\!\emph{3}$. Suppose $\sigma\in\PStab_{A'B}$ is a pure stabilizer Choi state on a bipartite system $A'|B$. By the result of \cite[Theorem~1]{Fattal_2004} there are some local Clifford unitaries $U_{A'}$ and $U_B$ that map $\sigma$ to a tensor product of local states $\ket{0}$ and a number of EPR pairs $\ket{\mathrm{Bell}}$. But since $\Tr_B\sigma = \chi_{A'}$ is maximally mixed, we conclude that
  \begin{equation}
    (U_{A'}\otimes U_B)\,\sigma\,(U_{A'}\otimes U_B)^\dag = \Omega^{\abs{A}}\otimes\proj{0}^{\abs{B}-\abs{A}}.
  \end{equation}
  This is the Choi state of a standard Clifford isometry $\tilde{V}:\Hil_A\to\Hil_B, \ket{\psi} \mapsto \ket{\psi}\ket{0}$ (see Table~\ref{tab:Choi_states}). Therefore the channel $\Phi$ can be represented as $\Phi[\rho] = V\rho V^\dag$ for $V = U_B \tilde{V} U_A$, where $U_A = U_{A'}^T$.
\end{proof}

Now, let us give a characterization to arbitrary non-adaptive Clifford channels.
\clearpage
\begin{theorem}[Clifford channels] \label{thm:non-adaptive_characterization}
  Let $\Phi : \TrClass(\Hil_A) \to \TrClass(\Hil_B)$ be a channel. The following statements are equivalent:
  \begin{enumerate}
    \item $\Phi$ is completely stabilizer preserving.
    \item Choi state $\sigma$ is mixed stabilizer.
    \item $\Phi$ may be realised by a non-adaptive stabilizer circuit.
  \end{enumerate}
\end{theorem}

\noindent Here we give a short proof of this Theorem~by constructiong a Clifford dilation of a channel. The proof will be strengthened in Section~\ref{sec:normal_form} to reveal the explicit structure of Clifford channels.

\begin{proof}
  As in the previous case, let us prove $\emph{2}\Rightarrow\!\emph{3}$. Suppose $\sigma\in\Stab_{A'B}$ is a mixed stabilizer Choi state on a bipartite system $A'|B$. Let us consider a minimal purification $\sigma_{A'BE} \in \PStab_{A'BE}$ to some environment $E$. Since $\Tr_{BE}\sigma_{A'BE} = \Tr_B\sigma = \chi_{A'}$, by Theorem~\ref{thm:unitary_characterization} the state $\sigma_{A'BE}$ corresponds to some isometry Clifford channel with $V : \Hil_A \to \Hil_{BE}$. Take $\Phi[\rho] = \Tr_E(V\rho V^\dag)$.
\end{proof}

A similar theorem was stated in \cite{Seddon_2019} while studying free operations for magic resources.
\begin{theorem}[\cite{Seddon_2019}] \label{thm:adaptive_characterization}
  Let $\Phi : \TrClass(\Hil_A) \to \TrClass(\Hil_B)$ be a channel. The following statements are equivalent:
  \begin{enumerate}
    \item $\Phi$ is completely adaptive stabilizer preserving.
    \item Choi state $\sigma$ is a mixture of stabilizer states.
  \end{enumerate}
\end{theorem}

\noindent One could conjecture that the part $\emph{3}$ also holds, i.e.\ these conditions imply the channel is adaptive Clifford. However, it turns out to be false: the authors of \cite{Heimendahl_2022} have explicitly shown that the class of channels that completely map stabilizer states to mixtures of stabilizer states is strictly larger than the class of channels represented by adaptive stabilizer circuits.

\section{Structure of Clifford channels} \label{sec:normal_form}

In this Section we further investigate the structure of circuits representing Clifford channels and find a simple normal form for such circuits. We discuss implications of the normal form to calculating information capacities of Clifford channels.

After releasing the first version of the preprint we found that the results of this Section were already stated in \cite[Section V]{Looi_2011} for prime local dimensions $d$. Nevertheless, we think it is useful to discuss it here, as this normal form seems to be under-recognised.

\subsection{Normal form of Clifford channels}

The proof of Theorem~\ref{thm:non-adaptive_characterization} relies on the analysis of mixed stabilizer Choi states. We can strengthen the analysis using the following Theorem:
\begin{theorem}[Remark from \cite{Bravyi_2006}] \label{thm:bipartite_mixstab_states}
  Every mixed stabilizer state $\rho\in\Stab_{AB}$ over bipartite system $A|B$ is locally Clifford equivalent to a product of chaotic states $\chi$ in both $A$ and $B$, local pure stabilizer states, maximally entangled states $\Omega$ and maxiamlly correlated states $\Sigma$.
\end{theorem}
\begin{proof}
  Let us purify the state $\rho$ to some state $\ket{\psi}\in\PStab_{ABE}$ with environment $E$. Following \cite[Theorem~5]{Bravyi_2006}, the state $\ket{\psi}$ is locally Clifford equivalent to a product of local stabilizer states, Bell states $\Omega$ between pairs of parties, and a number of $\ket{\mathrm{GHZ}}_{ABE}$ states over the whole tripartite system $A|B|E$. Taking partial trace over $E$:
  \begin{equation}
  \begin{gathered}
    \Tr_E\Omega_{AB} = \Omega_{AB}, \quad \Tr_E\Omega_{AE} = \chi_A, \quad \Tr_E\Omega_{BE} = \chi_B,\\
    \quad \Tr_E\proj{\mathrm{GHZ}}_{ABE} = \Sigma_{AB},
  \end{gathered}
  \end{equation}
  gives the expected result.
\end{proof}

\noindent We choose to refer to \cite[Theorem~5]{Bravyi_2006} because it clarifies the structure of purification. Surely, the proof can be obtained via simpler and more algorithmic methods, in a manner of \cite{Fattal_2004, Audenaert_2005}.

The corollary of Theorem~\ref{thm:bipartite_mixstab_states} is that any Clifford channel is equivalent to a product of stabilizer operations listed in Table~\ref{tab:Choi_states}.
\begin{theorem}[Normal form] \label{thm:normal_form}
  Let $\Phi : \TrClass(\Hil_A) \to \TrClass(\Hil_B)$ be a Clifford channel. Then it is Clifford unitary equivalent to a product of: initial state $\ket{0}$ preparations, chaotic state $\chi$ preparations, erasure channels $\Tr$, identity channels $\Id$ and dephasing channels $\Deph$.
\end{theorem}
\begin{proof}
  Consider the Choi state $\sigma\in\Stab_{A'B}$. By Theorem~\ref{thm:non-adaptive_characterization} and having in mind $\Tr_B\sigma = \chi_{A'}$, there exist Clifford unitaries $U_{A'}\in\CliffordGroup_{A'}$ and $U_B\in\CliffordGroup_B$ such that
  \begin{equation}
    (U_{A'}\otimes U_B)\,\sigma\,(U_{A'}\otimes U_B)^\dag = \proj{0}_B^a\otimes\chi_{B}^b\otimes\chi_{A'}^c\otimes\Omega_{A'B}^d\otimes\Sigma_{A'B}^e,
  \end{equation}
  where $a,b,c,d,e\in\Nat$ are some integers. This is the Choi state of a channel $\tilde{\Phi}$ composed of the listed elementary operations. Then $\Phi[\rho] = U_B \tilde{\Phi}[U_A \rho U_A^\dag] U_B^\dag$, where $U_A = U_{A'}^T$.
\end{proof}

\noindent So, any Clifford channel $\Phi$ has a decomposition of form
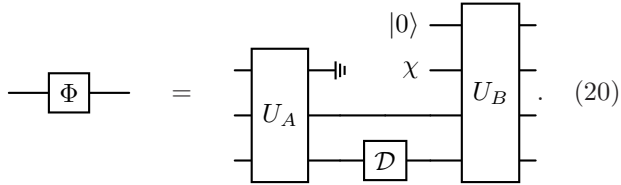
\begin{equation}
  \begin{quantikz}[wire types={q}, column sep={0.8cm,between origins}, align equals at=1.0]
    &\gate{\Phi} &
  \end{quantikz}
  \quad = \quad
  \begin{quantikz}[wire types={n,q,q,q}, row sep={0.6cm,between origins}, column sep={0.6cm,between origins}, align equals at=2.5]
    &              &[0.2cm]   &                &\lstick{$\ket{0}$} &[0.2cm]\gate[4]{U_B}\setwiretype{q} & \\
    &\gate[3]{U_A} &\ground{} &\setwiretype{n} &\lstick{$\chi$}    &\setwiretype{q}              & \\
    &              &          &                &                   &                             & \\
    &              &          &\gate{\Deph}    &                   &                             & \\
  \end{quantikz}
  .
\end{equation}
This normal form has freedom in permutations of wires and in the choice of $U_A$ and $U_B$, while the numbers $N_\Id$ of identity qubit channels and $N_\Deph$ of qubit dephasing channels are constant. Given such decomposition, we will call Clifford unitary $U_A$ the \emph{encoding gate} and $U_B$ the \emph{decoding gate}.

\begin{remark} \label{rmk:other_research}
  A closely related normal form was independently obtained in \cite{Kliuchnikov_2023}. The authors proved that all stabilizer circuits with affine classical computations and Pauli control can be decomposed into: unitary encoding and decoding gates, one-bit measurements, initial state $\ket{0}$ preparations, preparations of random bits, affine Booleam operations and classical control over Pauli gates. As mentioned in Remark~\ref{rmk:allowed_control}, any such circuit represents a Clifford channel, so the careful analysis of normal form from Theorem~\ref{thm:normal_form} should lead to the result of \cite{Kliuchnikov_2023}.
\end{remark}

\subsection{Information-theoretic properties of Clifford channels}

Suppose we want to solve some data transmission problem using a channel $\Phi$ that connects a sender $A$ to a reciever $B$. Different tasks can be considered, depending on objectives or available resources of the communicating parties. Three most commonly studied tasks are:
\begin{enumerate}
  \item The task of communicating classical bit messages from the sender $A$ to the reciever $B$ by multiple uses of a quantum channel $\Phi$. The rate of communicated bits per channel use can be quantified for any communication protocol, and the least upper bound on communication rates in all possible scenarios is given by \emph{classical capacity} $C(\Phi)$.
  \item The task of communicating classical bit messages from the sender $A$ to the reciever $B$ by using channel $\Phi$ and unlimited shared resource of entanglement between $A$ and $B$. The maximal rate of communication with allowed entanglement is characterised by \emph{entanglement-assisted classical capacity} $C_{ea}(\Phi)$.
  \item The task of communicating multiqubit quantum states (i.e. quantum data) from the sender $A$ to the reciever $B$ by multiple uses of a quantum channel $\Phi$. In this case the \emph{quantum capacity} $Q(\Phi)$ describes the optimal communication rate.
\end{enumerate}
For detailed introduction to the field of quantum Shannon theory and various data transmission protocols consult \cite{Holevo_2019, Wilde_2017}.

Let us show how using the decomposition from Theorem~\ref{thm:normal_form} we can easily solve these problems for a Clifford channel $\Phi$. Suppose that we found a normal form of the channel $\Phi$ with encoding gate $U_A$, decoding gate $U_B$, having $N_{\Id}$ identity channels and $N_{\Deph}$ dephasing channels in the decomposition. The sender and the reciever can perform operations $U_A^\dag$ and $U_B^\dag$ before and after using the channel $\Phi$, so it is sufficient to consider the problem of communicating through a tensor product of qubit discardings, state preparations, identity channels and dephasing channels. State preparations and qubit discardings do not allow any data transmission, so should be omitted. Using one-qubit dephasing channel $\Deph$, one can transmit one bit of classical information by preparing and measuring $Z$-bases state ($C(\Deph)=1$), but it breaks entanglement ($C_{ea}(\Deph)=1$) and destroys quantum data ($Q(\Deph)=0$). Using one-qubit identity channel $\Id$, one can transmit one classical bit ($C(\Id)=1$) or two classical bits using entanglement and superdense coding protocol ($C_{ea}(\Id)=2$), it also transmits one qubit of data ($Q(\Id)=1$). The qubit channels $\Id$ and $\Deph$ are especially simple and their capacities are additive. Thus, the information-theoretic capacities of a Clifford channel $\Phi$ are
\begin{equation}
  C(\Phi) = N_{\Id} + N_{\Deph},\quad C_{ea}(\Phi) = 2 N_{\Id} + N_{\Deph},\quad Q(\Phi) = N_{\Id}.
\end{equation}
Similarly, one can evaluate other interesting capacities of Clifford channels \cite{Holevo_2019, Wilde_2017}. Since $U_A^\dag$ and $U_B^\dag$ are Clifford unitaries, the states for optimal encoding can be chosen to be stabilizer, which proves the stabilizer theory analogue of Gaussian optimizers conjecture \cite{Giovannetti_2014, Holevo_2015, Holevo_2019}. Let us note that the connected results were announced (yet not completed) in \cite{Fattal_2007}.

\section{Stabilizer preserving channels are Clifford} \label{sec:stabilizer_preserving}

In this Section prove that if a channel is stabilizer preserving, then it is also completely stabilizer preserving and Clifford. Additionally, we characterise Clifford channels in terms of their dual channels.

We have given a characterization for Clifford channels in terms of \emph{completely} stabilizer preserving property in Section~\ref{sec:characterization}. Now we want to loosen this condition. For adaptive Clifford channels the ``completely'' condition is necessary \cite{Seddon_2019}. However, in linear optics if a channel is Gaussian-to-Gaussian, then it is also completely Gaussian-to-Gaussian \cite{De_Palma_2015}. By analogy between stabilizer and linear optical circuits (see Appendix~\ref{appendix:analogy}) it is natural to assert that the same holds for Clifford case.

Before proving that all stabilizer preserving channels are Clifford, let us consider two useful Lemmata. They are dual to each other and expain the relation between stabilizer states and Pauli observables.

\begin{lemma} \label{lem:pauli_then_stabilizer}
  A state $\rho\in\TrClass(\Hil^n)$ is mixed stabilizer if and only if $\Tr(\rho P) \in \{-1,0,+1\}$ for all Pauli observables $P\in\PauliGroup^n$.
\end{lemma}
\begin{proof}
  Pauli matrices $\PauliGroup^n$ without signs form an orthogonal basis, so any decomposition of states into a sum of Pauli operators is unique. Mixed stabilizer states are decomposed as in Equation~(\ref{eq:mixed_stabilizer}), therefore for any mixed stabilizer state $\rho\in\MixStab^n$ the coefficients $\Tr(P \rho)$ lie in $\{-1,0,+1\}$.

  Let us prove the converse. If for all non-trivial Pauli observables $P$ it holds $\Tr(\rho P) = 0$, then $\rho$ is a chaotic state $\chi^n$. Otherwise, there exists some Pauli observable $P$ such that $\Tr(\rho P) = 1$. We can choose a Clifford unitary $U\in\CliffordGroup^n$ so that $U^\dag P U = Z_1$, thus $\Tr(U\rho U^\dag Z_1) = 1$. That means
  \begin{equation}
    U\rho U^\dag = \proj{0}\otimes\rho',
  \end{equation}
  where $\rho'$ is some ($n-1$)-qubit state, for which the assumption of Lemma still holds. By induction on the number of qubits we conclude that $\rho$ is mixed stabilizer.
\end{proof}

Let us consider one-bit measurement channels $\mathcal{M} : \TrClass(\Hil^n) \to \ell_1(\Int_2)$. By Theorem~\ref{thm:normal_form} every Clifford one-bit measurement channel is either a preparation of mixed stabilzer state or a Pauli measurement. Stabilizer-preserving one-bit measurement channel $\mathcal{M}$ is the channel that maps pure stabilizer states $\PStab^n$ to uniform distributions over affine subspaces of $\Int_2$, i.e.\ deterministic bit states $0$ and $1$ or a uniformly random bit. This channel is described by a pair of effects $\{E_0,E_1\}$. Define $Q = E_0 - E_1$, then $-I\leq Q \leq +I$. Let us show that the channel $\mathcal{M}$ is a Clifford one-bit measurement, i.e.\ $Q$ is either a constant observable $\{-I,0,+I\}$ or non-trivial Pauli observable.

\begin{lemma} \label{lem:stabilizer_then_pauli}
  Let $Q, -I\leq Q\leq I$ be an observable over $\Hil^n$ such that for any pure stabilizer state $\ket{\psi}\in\Stab^n$ it holds $\bra{\psi}Q\ket{\psi}\in\{-1,0,+1\}$. Then $Q$ is either zero or a Pauli observable.
\end{lemma}
\begin{proof}
  Firstly, we will show that the assumption of Lemma holds for all $I^m\otimes Q$, where $m\in\Nat$ represents some qubit environment. Let us add one-qubit environment $m=1$ and let $\ket{\psi}\in\PStab^{n+1}$ be a pure stabilizer state. The case of separable $\ket{\psi}$ is trivial, so suppose $\ket{\psi}$ is entangled with environment. Using the result of \cite[Theorem~1]{Fattal_2004}, there exists some $U\in\CliffordGroup^n$ such that $(I\otimes U)\ket{\psi} = \ket{\mathrm{Bell}}\ket{0^{n-1}}$. Denote $\tilde{Q} = U^\dag Q U$, then
  \begin{equation}
    \bra{\psi}I\otimes Q\ket{\psi}
    = \Tr^{n+1}(\Omega\otimes\proj{0}^{n-1}\cdot I\otimes\tilde{Q})
    = \Tr^n(\chi\otimes\proj{0}^{n-1} \cdot \tilde{Q}).
  \end{equation}
  The chaotic state $\chi$ can be decomposed in two different ways:
  \begin{equation}
  \begin{aligned}
    \bra{\psi}I\otimes Q\ket{\psi}
    &= \frac{1}{2} \bra{0\,0^{n-1}}\tilde{Q}\ket{0\,0^{n-1}} + \frac{1}{2} \bra{1\,0^{n-1}}\tilde{Q}\ket{1\,0^{n-1}} \\
    &= \frac{1}{2} \bra{+\,0^{n-1}}\tilde{Q}\ket{+\,0^{n-1}} + \frac{1}{2} \bra{-\,0^{n-1}}\tilde{Q}\ket{-\,0^{n-1}}.
  \end{aligned}
  \end{equation}
  By this decomosition we see that
  \begin{equation}
    \bra{\psi}I\otimes Q\ket{\psi} \in \{-1,-\frac{1}{2},0,+\frac{1}{2},+1\}.
  \end{equation}
  Let us prove that it cannot equal $\pm\frac{1}{2}$. Suppose $\bra{0\,0^{n-1}}\tilde{Q}\ket{0\,0^{n-1}} = 1$ and $\bra{1\,0^{n-1}}\tilde{Q}\ket{1\,0^{n-1}} = 0$. The state $\ket{0^n}$ maximises the norm of $\tilde{Q}$, so it is an eigenstate. Similarly, we also see that either $\ket{+\,0^{n-1}}$ or $\ket{-\,0^{n-1}}$ is eigenstate with eigenvalue $1$. But that implies the $1$-eigenspace also contains $\ket{1\,0^{n-1}}$, a contradiction. So, $\bra{\psi}I\otimes Q\ket{\psi}\in\{-1,0,+1\}$. By induction we make the same conclusion for all $m\in\Nat$.

  We conclude, by considering purifications, that $\Tr(\rho Q)\in\{-1,0,+1\}$ for all mixed stabilizer states $\rho\in\Stab^n$. Let us look at the action of $Q$ on chaotic state $\chi^n$. If $\Tr(\chi^n Q) = 1$, then the Cauchy–Schwarz inequality $\abs{\frac{1}{2^n}\Tr(I\, Q)}\leq 1$ is maximized, so $Q=I$. Similarly, if $\Tr(\chi^n Q) = -1$, then $Q=-I$. Next, suppose $\Tr(\chi^n Q) = 0$. In this case
  \begin{equation}
    \Tr\left(\frac{I\pm P}{2^n}Q\right) = \frac{1}{2^n}\Tr(P Q) \in \{-1, 0, +1\}
  \end{equation}
  for all non-trivial Pauli observables $P\in\PauliGroup^n$. If $\frac{1}{2^n}\Tr(P Q) = 0$ for all $P$, then $Q=0$. Otherwise, there exists some non-trivial Pauli observable $P\in\PauliGroup^n$ such that $\frac{1}{2^n}\Tr(P Q) = 1$. That relation maximizes Cauchy–Schwarz inequality, so $Q=P$.
\end{proof}

We can characterise stabilizer preserving property in terms of the action of the dual channel $\Phi^*$ on a Pauli group. That may resemble the definition of quasifree maps between algebras of canonical commutation relations \cite{Petz_1990}.
\begin{theorem} \label{thm:dual_stabilizer_preserving}
  A channel $\Phi : \TrClass(\Hil^n) \to \TrClass(\Hil^m)$ is stabilizer preserving if and only if the dual channel $\Phi^* : \Bounded(\Hil^m) \to \Bounded(\Hil^n)$ preserves Pauli group with zero: for all Pauli observables $P\in\PauliGroup^m$ the observable $\Phi^*[P]$ is Pauli or zero.
\end{theorem}
\begin{proof}
  By Lemma~\ref{lem:pauli_then_stabilizer} the stabilizer preserving condition for $\Phi$ means exactly that for all pure stabilizer states $\ket{\psi}\in\PStab^n$ and Pauli observables $P\in\PauliGroup^m$ it holds $\Tr(\Phi[\proj{\psi}] P) \in \{-1,0,+1\}$. That rewrites to $\bra{\psi}\Phi^*[P]\ket{\psi} \in \{-1,0,+1\}$, which by Lemma~\ref{lem:stabilizer_then_pauli} holds if and only if $\Phi^*[P]$ is a Pauli observable or zero.
\end{proof}

Now, we can easily prove the claim of interest.
\begin{theorem} \label{thm:stabilizer_preserving_are_clifford}
  Let $\Phi : \TrClass(\Hil_A) \to \TrClass(\Hil_B)$ be a stabilizer preserving quantum channel. Then $\Phi$ is a Clifford channel.
\end{theorem}
\begin{proof}
  We want to prove that if $\Phi$ is stabilizer preserving, then $\Id_E \otimes \Phi$ is also stabilizer preserving for any qubit environment $E$. Fix some Pauli observable $P_{EB}\in\PauliGroup_{EB}$ and represent it as a product $P_{EB} = P_E\otimes P_B$ of Pauli observables $P_E\in\PauliGroup_E$ and $P_B\in\PauliGroup_B$, then $\Id_E\otimes\Phi^*[P_{EB}] = P_E\otimes\Phi^*[P_B]$. By Theorem~\ref{thm:dual_stabilizer_preserving}, the observable $\Phi^*[P_B]$ is zero or Pauli, so $P_E\otimes\Phi^*[P_B]$ is also either a Pauli observable or zero. That means $\Id_E \otimes \Phi$ is stabilizer preserving.
\end{proof}

We conclude that in Theorem~\ref{thm:non-adaptive_characterization} the usual stabilizer preserving property is sufficient. Let us formulate the result for some interesting cases.

\begin{corollary} \label{cor:state_preparation_channels}
  If $\ell_1(\Int_2^n)\to\TrClass(\Hil^m)$ is a state preparation channel that maps uniform distributions over affine subspaces of $\Int_2^n$ to mixed stabilizer states, then it corresponds to stabilizer circuit preparation of mixed stabilizer states.
\end{corollary}

\begin{corollary} \label{cor:measurement_channels}
 If $\TrClass(\Hil^n)\to\ell_1(\Int_2^m)$ is a $m$-bit measurement channel that maps pure stabilizer states to uniform distributions over affine subspaces, then it outputs the results of Pauli measurements, with possible post-processing by affine Boolean maps and random bit inputs.
\end{corollary}

\begin{corollary} \label{cor:classical_channels}
 If $\ell_1(\Int_2^n)\to\ell_1(\Int_2^m)$ is classical channel that maps uniform distributions over affine subspaces of $\Int_2^n$ to uniform distributions over affine subspaces of $\Int_2^m$, then it can be represented by some affine Boolean circuit.
\end{corollary}

Finally, let us comment on the structure of channels that send pure stabilizer states to pure stabilizer states. Except for Clifford isometry channels, there exists an example of \emph{state reset channel} $\rho \mapsto \proj{\psi}\Tr(\rho)$ that discards all qubits and prepares a pure stabilizer state $\ket{\psi}\in\PStab_B$. By Theorem~\ref{thm:stabilizer_preserving_are_clifford} it follows there are no other options.

\section{Conclusion and outlook} \label{sec:conlusion}

We have shown that the notion of a Clifford channel can be naturally defined in a number of equivalent ways: in terms of (complete) stabilizer preserving property, in terms of the stabilizerness of Choi state, in terms of the dual channel action on Pauli group, in terms of Clifford dilation or in terms of realization via non-adaptive stabilizer circuits. The channel is Clifford isometry if it is completely pure stabilizer preserving, or if its Choi state is pure stabilizer, or if it is represented by an isometry stabilizer circuit. Complete stabilizer preserving property follows already from the usual stabilizer preserving property. If the channel maps pure stabilizer states to pure stabilizer states, it may be either a state reset channel or a Clifford isometry. Clifford channels have particularly simple structure: up to Clifford unitary encoding and decoding operations, they can be represented as a product of qubit discardings, initial state $\ket{0}$ and chaotic state $\chi$ preparations, identity channels and dephasing channels. That leads to their simple information-theoretic structure: channel capacities depend only on the number of identity and dephasing channels in the normal form.

There is a number of directions the work can be generalised. First of all, one can prove the statements for qudit systems of local dimensions $d>2$ \cite{Beaudrap_2013, Gheorghiu_2014}. Some job in this direction was done in \cite{Looi_2011}. Our proofs seem to mainly rely on dimension-independent aspects of stabilizer theory, so for all prime local dimensions $d$ they are easily replicated. For composite $d$ the structure of Clifford channels is more sophisticated.

Next, one can further study adaptive Clifford channels. It would be interesting to better understand the distinction between completely adaptive stabilizer preserving channels and adaptive Clifford channels \cite{Heimendahl_2022}. It would then be interesting to characterise extreme points in the space of adaptive Clifford channels. The results of \cite[Theorem~4]{Heimendahl_2022} show that such points can be represented by a sequence of adaptive Pauli measurements followed by Clifford channels depending on measurement outcomes. Using the deferred measurement principle on such realization we can conclude that extreme adaptive Clifford channels lie in finite levels of Clifford hierarchy. As suggested in Section~\ref{sec:characterization}, it may be useful to study higher Clifford hierarchy channels, even though the structure of unitary Clifford hierarchy is underexplored \cite{Beigi_2010, Cui_2017, Rengaswamy_2019, Pllaha_2020, De_Silva_2025}.

We expect that adaptive Clifford channels should differ from classical channels exactly by the possibility to use the resource of entanglement, but not in other resources. Following this expectation, we conjecture that adaptive Clifford channels should satisfy some information-theoretic conjectures, such as additivity of classical capacity. Arbitrary quantum channels can be obtained from adaptive Clifford channels by applying the resource of magic \cite{Seddon_2019}. It would be interesting to assess the effect of magic on information transmission (see related work \cite{Bu_2024, Oliviero_2022, Cao_2024}). Low-magic channels should be understood as close-to-Clifford, so possibly one could apply some results of quantum Boolean analysis \cite{Montanaro_2008, Rouz_2022}.

Finally, let us note that the operational interpretation of Clifford channels may have applications to resource theories \cite{Chitambar_2019, Gour_2024}. By analogy to the non-convex resource theory of non-Gaussianity \cite{Albarelli_2018, Lami_2018}, one can define a resource theory of non-stabilizerness with mixed stabilizer states as free states and Clifford channels as free operations. Such resource theory is different from the convex resource theory of magic \cite{Veitch_2012, Seddon_2019}, which is analogous to the theory of Wigner negativity in Gaussian case. For classical systems, Clifford operations are affine Boolean maps with possible uniformly random bit inputs, so the non-stabilizerness can be seen as the amount of multiplicativity. It may therefore have a direct relation to the multiplication cost of functions, often studied in algebraic complexity.

\section*{Acknowledgements}
The authors are grateful to A.S.~Holevo, G.G.~Amosov and E.O.~Kiktenko for useful comments. This work was supported by the Russian Science Foundation under grant no. 24-11-00145, \url{https://rscf.ru/project/24-11-00145/}.

\nocite{Kay_2023} 
\bibliography{bibliography}

\appendix
\section{The analogy between stabilizer and Gaussian theory} \label{appendix:analogy}

In this Appendix we formulate an analogy between theories of bosonic and stabilizer circuits. We give it in Table~\ref{tab:analogy} with three columns corresponding to stabilizer theories over qubit circuits \cite{Gottesman_1997}, odd-dimensional qudit circuits \cite{Gottesman_1999_2} and linear bosonic circuits \cite{Weedbrook_2012}. One might as well consider CSS-preserving rebit circuits theory \cite{Delfosse_2015} or fermionic linear optics \cite{Terhal_2002, Jozsa_2008, Reardon_Smith_2023, Dias_2023}. Some of important concepts in Gaussian theory can only be understood approximately or in terms of Shwartz distributions, we write ``approx.'' in this case. We do not aim this list to be exhaustive, and the definitions are found across the literature.

\begin{table*}[h]
\caption{The analogy between concepts of qubit and qudit stabilizer theories and bosonic linear optics.}
\label{tab:analogy}
\begin{ruledtabular}
\begin{tabular}{ccc}
  Qubit stabilizer theory & Qudit stabilizer theory & Bosonic linear optics theory \\
  \hline
  $n$ qubits & $n$ qudits & $n$ bosonic modes \\
  configuration space $\Int_2^n$ & configuration space $\Int_d^n$ & configuration space $\Real^n$ \\
  phase space $\Int_2^{2n}$ & phase space $\Int_d^{2n}$ & phase space $\Real^{2n}$ \\
  phase factor $-1$ & phase factor $\omega = e^{i\,2\pi/d}$ & phase factor $e^{i\,2\pi t}$ \\
  $Z$-basis $\{\ket{0},\ket{1}\}$ & position states $\{\ket{0},\dots,\ket{d-1}\}$ & Schr{\"o}dinger position representation \\
  $X$-basis $\{\ket{+},\ket{-}\}$ & momentum states $\{\frac{1}{\sqrt{d}}\sum_{j=0}^{d-1} \omega^{k j}\ket{j}\}_k$ & Schr{\"o}dinger momentum representation \\
  Pauli group & Heisenberg-Weyl group & Heisenberg group \\
  Pauli operators & Heisenberg-Weyl operators & displacement operators \\
  $X$ gate & shift operator $X\ket{j} = \ket{j+1}$ & position shift $e^{-i q \hat{p}}$ \\
  $Z$ gate & clock operator $Z\ket{j} = \omega^j\ket{j}$ & momentum shift $e^{i p \hat{q}}$ \\
  $ZX = -XZ$ & $ZX = \omega XZ$ & Weyl commutation relations \\
  Clifford group & Clifford group & affine symplectic group \\
  Hadamard gate $H$ & Fourier transform $F\ket{k} = \frac{1}{\sqrt{d}}\sum_{j=0}^{d-1} \omega^{k j}\ket{j}$ & Fourier transform \\
  none \cite{Schmid_2022} (or \cite{Wootters_1987}) & discrete Wigner function \cite{Gross_2006} & Wigner function \cite{Wigner_1932} \\
  pure stabilizer states & pure stabilizer states & pure Gaussian states \\
  mixed stabilizer states & mixed stabilizer states & mixed Gaussian states \\
  Bell state $\ket{\mathrm{Bell}}$ & maximally entangled state & approx. maximally entangled state \\
  Choi state & Choi state & approx. Choi state \\
  $Z$-measurement & $Z$-measurement & homodyne position measurement \\
  Clifford channels & Clifford channels & Gaussian channels \\
  classical bit $\Int_2$ & classical dit $\Int_d$ & classical variable $\Real$ \\
  Boolean operations & polynomial operations & polynomial operations \\
  affine Boolean maps & affine maps & affine maps \\
  classical Pauli control & classical Heisenberg-Weyl control & controlled displacements
\end{tabular}
\end{ruledtabular}
\end{table*}

\end{document}